%% file: polyspace-subsetsum.tex
\documentclass[11pt]{article}
\usepackage{hyperref}
\usepackage{amsfonts,amsmath,amsthm,wrapfig}
\usepackage{comment,fullpage}
\usepackage[noend]{algorithmic}
\usepackage[boxed]{algorithm}
\usepackage{xspace}
\usepackage{framed}
\usepackage{xcolor}
\usepackage{verbatim}
\usepackage{tikz}
\usepackage{paralist}
\usetikzlibrary{arrows}

\newtheorem{theorem}{Theorem}
\newtheorem{lemma}{Lemma}[section]
\newtheorem{definition}{Definition}[theorem]

\newtheorem{observation}{Observation}[theorem]
\newtheorem{corollary}[theorem]{Corollary}

\newtheorem*{rep@theorem}{\rep@title}
\newcommand{\newreptheorem}[2]{%
\newenvironment{rep#1}[1]{%
 \def\rep@title{#2 \ref{##1}}%
 \begin{rep@theorem}[restated]}%
 {\end{rep@theorem}}}

\newreptheorem{theorem}{Theorem}
\newreptheorem{corollary}{Corollary}
\newreptheorem{proposition}{Proposition}

\newcommand{\poly}{\ensuremath{\mathrm{poly}}\xspace}
\newcommand{\polylog}{\ensuremath{\mathrm{polylog}}\xspace}

\newcommand{\algorithmiccommentt}[1]{\colorbox{black!10}{#1}}
\newcommand{\LineIf}[2]{ \STATE \algorithmicif\ {#1}\ \algorithmicthen\ {#2} }
\newcommand{\LineIfElse}[3]{ \STATE \algorithmicif\ {#1}\ \algorithmicthen\ {#2} \algorithmicelse\ {#3} }

\newcommand\E{\mathbb{E}}

\begin{document}
\author{Nikhil Bansal\thanks{Eindhoven University of Technology, Netherlands.
\texttt{n.bansal@tue.nl}.
Supported by a NWO Vidi grant 639.022.211 and an ERC consolidator grant 617951.
 } \and Shashwat Garg\thanks{Eindhoven University of Technology, Netherlands.  
\texttt{s.garg@tue.nl}.
Supported by the Netherlands Organisation for Scientific Research (NWO) under project no.~022.005.025.} 
 \and Jesper Nederlof\thanks{Eindhoven University of Technology, The Netherlands.\texttt{j.nederlof@tue.nl}. Supported by NWO Veni grant 639.021.438. }
 \and Nikhil Vyas\thanks{Indian Institue of Technology Bombay, India. \texttt{vyasnikhil96@gmail.com}}
 }
\title{
Faster Space-Efficient Algorithms for Subset Sum, k-Sum and Related Problems
}

\maketitle


\begin{abstract}
We present 
randomized algorithms that solve Subset Sum and Knapsack instances with $n$ items in $O^*(2^{0.86n})$ time, where the $O^*(\cdot)$ notation suppresses factors polynomial in the input size, and polynomial space, assuming random read-only access to exponentially many random bits. 
These results can be extended to solve Binary Linear Programming on $n$ variables with few constraints in a similar running time.
We also show that for any constant $k\geq 2$, random instances of $k$-Sum can be solved using $O(n^{k-0.5}\polylog(n))$ time and $O(\log n)$ space, without the assumption of random access to random bits.

Underlying these results is an algorithm that determines whether two given lists of length $n$ with integers bounded by a polynomial in $n$ share a common value. Assuming random read-only access to random bits, we show that this problem can be solved using $O(\log n)$ space significantly faster than the trivial $O(n^2)$ time algorithm if no value occurs too often in the same list.

\end{abstract}

\clearpage
\input{intro}

\input{preliminaries}

\input{listdisjointness}

\input{subsetsum}

\input{randomksum}

\input{conclusions}

\normalsize
\bibliographystyle{plain}
\bibliography{ref}

\end{document}

%% file: intro.tex
\section{Introduction}

The Subset Sum problem and the closely related Knapsack problem are two of the most basic NP-Complete problems. 
In the Subset Sum problem we are given integers $w_1,\ldots,w_n$ and an integer target $t$ and are asked to find a subset $X$ of indices such that $\sum_{i \in X} w_i=t$. In the Knapsack problem we are given integers $w_1,\ldots,w_n$ (weights), integers $v_1,\ldots,v_n$ (values) and an integer $t$ (weight budget) and want to find a subset $X$ maximizing $\sum_{i \in X}v_i$ subject to the constraint $\sum_{i \in X}w_i \leq t$.

It is well known that both these problems can be solved in time $O^*(\min(t,2^n))$, where we use the $O^*(\cdot)$ notation to suppress factors polynomial in the input size. In their 
landmark paper introducing the Meet-in-the-Middle approach,  Horowitz and Sahni~\cite{DBLP:journals/jacm/HorowitzS74} solve both problems in $O^*(2^{n/2})$ time and $O^*(2^{n/2})$ space, substantially speeding up the trivial $O^*(2^{n})$ time polynomial space algorithm. Later this was improved to $O^*(2^{n/2})$ time and $O^*(2^{n/4})$ space by Schroeppel and Shamir~\cite{DBLP:journals/siamcomp/SchroeppelS81}. While these approaches have been extended to obtain improved tradeoffs in several special cases, particularly if the instances are random or ``random-like"~\cite{DBLP:conf/icalp/AustrinKKM13,DBLP:conf/stacs/AustrinKKN15,DBLP:conf/stacs/AustrinKKN16,DBLP:conf/crypto/DinurDKS12,DBLP:conf/eurocrypt/Howgrave-GrahamJ10}, the above algorithms are still the fastest known for general instances in the regimes of unbounded and polynomial space.

The idea in all these algorithms is to precompute and store an exponential number of intermediate quantities in a lookup table, and use this to speed up the overall algorithm. In particular, in the Meet-in-the-Middle approach for Subset sum, one splits the $n$ numbers into two halves $L$ and $R$ and computes and stores the sum of all possible $2^{n/2}$ subsets of $L$, and similarly for $R$. Then for every partial sum $a$ of a subset of $L$ one looks up, using binary search, whether there is the partial sum $t-a$ for a subset of $R$. As these approaches inherently use exponential space, an important long-standing question (e.g.~\cite{DBLP:conf/icalp/AustrinKKM13},~\cite{openproblems}, Open Problem 3b in \cite{DBLP:conf/iwpec/Woeginger04}, Open Problem 3.8b in \cite{DBLP:conf/iwpec/2004}) has been to find an algorithm that uses polynomial space and runs in time $O^*(2^{(1-\varepsilon)n})$ for some $\varepsilon >0$. We make some progress in this direction.
 
\subsection{Our Results}\label{subsec:ourres}
\paragraph{Subset Sum, Knapsack and Binary Linear Programming:}
In this paper we give the first polynomial space randomized algorithm for Subset Sum, Knapsack and Binary Linear Programming (with few constraints) that improves over the trivial $O^*(2^n)$ time algorithm for worst-case inputs, under the assumption of having random read-only access to random bits. Our main theorem reads as follows:
\begin{theorem}\label{thm:sss}
	There are Monte Carlo algorithms solving Subset Sum and Knapsack using
	$O^*(2^{0.86n})$ time and polynomial space. The algorithms assume random read-only access to exponentially many random bits.
\end{theorem}

In the \emph{Binary Linear Programming} problem we are given $v,a^1,\ldots,a^d \in \mathbb{Z}^n$ and $u_1,\ldots,u_d \in \mathbb{Z}$ and are set to solve the following optimization problem:
\[
\begin{aligned}
& \text{minimize} & & \langle v,x \rangle \\
& \text{subject to}
& & \langle a^j, x \rangle \leq u_j, \text{ for $j=1,\ldots,d$}\\
&&& x_i \in \{0,1\}, \text{ for $i=1,\ldots,n$}.
\end{aligned}
\]
Using reductions from previous work, we obtain the following consequence of Theorem~\ref{thm:sss}:
\begin{corollary}\label{thm:bip}
	There exists a Monte Carlo algorithm which solves Binary Linear Programming instances with maximum absolute integer value $m$ and $d$ constraints in polynomial space and time $O^*(2^{0.86n}(\log(mn)n)^{O(d)})$. The algorithm assumes random read-only access to exponentially many random bits.
\end{corollary}

Though read-only access to random bits is a rather uncommon assumption in algorithm design,\footnote{The only exception we are aware of is~\cite{beame2013element} where it is treated as a rather innocent assumption.} such an assumption arises naturally in space-bounded computation (where the working memory is much smaller than the number of random bits needed), and has received more attention in complexity theory (see e.g.~\cite{DBLP:journals/combinatorica/Nisan92,DBLP:journals/tcs/Nisan93,DBLP:conf/coco/Saks96}).

From existing techniques it follows that the assumption is weaker than assuming the existence of sufficiently strong pseudorandom generators. We elaborate more on this in Section~\ref{sec:conc} and refer the reader to~\cite{DBLP:conf/coco/Saks96} for more details on the different models for accessing input random bits.


\paragraph{List Disjointness:} A key ingredient of Theorem~\ref{thm:sss} is an algorithmic result for the {\em List Disjointness} problem.
In this problem one is given (random read-only access to) two lists $x, y \in [m]^n$ with $n$ integers in range $[m]$ with $m\leq \poly(n)$, and the goal is to determine if there exist indices $i, j$ such that $x_{i} = y_{j}$. This problem can be trivially solved by sorting using $\tilde{O}(n)$ time\footnote{The $\tilde{O}(\cdot)$ notation suppresses factor polylogarithmic in the main term.} and $O(n)$ space, or in $\tilde{O}(n^2)$ running time and $O(\log n)$ space using exhaustive search. 
We show that improved running time bounds with low space can be obtained, if the lists do not have too many ``repeated entries".

\begin{theorem}\label{thm:LD-ST}
There is a Monte Carlo algorithm that takes as input an instance $x,y \in [m]^n$ of List Disjointness with $m \leq \poly(n)$, an integer $s$, and an upper bound $p$ on the sum of the squared frequencies of the entries in the lists, i.e. $\sum_{v=1}^m |x^{-1}(v)|^2+|y^{-1}(v)|^2 \leq p$, and solves the List Disjointness instance $x,y$ in time $\tilde{O}(n\sqrt{p/s})$ and $O(s \log n)$ space for any $s \leq n^2/p$. The algorithm assumes random read-only access to a random function $h:[m] \rightarrow [n]$.
\end{theorem}
Note the required function $h$ can be simulated using polynomially many random bits.
This result is most useful when $s=O(1)$, and it gives non-trivial bounds when the number of repetitions is not too large, i.e. we can set $p \ll n^2$.
In particular, if $m=\Omega(n)$ and the lists $x$ and $y$ are ``random-like'' then we expect $p=O(n)$ and we get that List Disjointness can be solved in $\tilde{O}(n^{3/2})$ time and $O(\log n)$ space.
Moreover, in this random-like setting, as $s$ increases to $n$, the running time of our algorithm smoothly goes to $\tilde{O}(n)$, matching the sorting based running time.

As List Disjointness arises as a fundamental subroutine in many algorithms (especially those based on the Meet-in-the-Middle approach),  Theorem~\ref{thm:LD-ST} may also be useful for other problems.

\paragraph{k-Sum:} We illustrate that the random read-only access assumption is not needed when an instance has sufficient random properties itself. In particular, we apply the techniques behind Theorem \ref{thm:LD-ST} to solve random instances of the $k$-Sum problem. In this problem we are given $k$ lists $w^1,\ldots,w^k$ each with $n$ integers and an integer $t$ and the goal is to find some $s_i \in w^i$ for each $i$ such that $\sum_{i=1}^n s_i = t$. We solve random instances of this problem in the following precise sense:
\begin{theorem}\label{thm:randksum}
 There is a randomized algorithm that given a constant $k \geq 2$, $k$ uniformly random vectors $w^1,\ldots,w^k \in_R [m]^n$, with $n \leq m \leq \poly(n)$ and $m$ being a multiple of $n$, and a target $t$ (that may be adversarially chosen depending on the lists), finds $s_i \in w^i$ satisfying $\sum_{i=1}^k s_i=t$ if they exist with constant probability using $\tilde{O}(n^{k-0.5})$ time and $O(\log n)$ space.
\end{theorem}
In particular, this implies an $\tilde{O}(n^{1.5})$ time algorithm using $O(\log n)$ space for the random $2$-Sum problem (without any read-only random bits).
Note that unlike algorithms for worst-case $k$-Sum, we cannot reduce $k$-Sum to $\ell$-Sum on $n^{\lceil k / \ell \rceil}$-sized lists in the usual way as the produced lists will not be random.
By applying Theorem~\ref{thm:LD-ST} we also show how random instances of $k$-Sum for even constant $k$ can be solved in $\tilde{O}(n^{3k/2})$ time and $O(\log n)$ space assuming random read-only access to random bits. See Section~\ref{sec:ksum} for more details.


%
%


\subsection{Our Main Techniques}
\paragraph{Techniques behind Theorem~\ref{thm:sss}:}
To obtain the algorithm for Subset Sum in Theorem~\ref{thm:sss}, we apply the Meet-in-the-Middle approach from~\cite{DBLP:journals/jacm/HorowitzS74} as outlined above, but use Theorem~\ref{thm:LD-ST} instead of using a lookup table as we cannot store the latter in memory. While we can also not store the two lists of partial sums in memory, we can simulate random access to it as a fixed partial sum in these lists is easily computed in $O(n)$ time.
As the running time in Theorem~\ref{thm:LD-ST} depends on the parameter $p$, which we may need to set as large as $2^n$ in general,
this does not directly give us an $O^*((2-\varepsilon)^n)$ time polynomial space algorithm for $\varepsilon > 0$.

To obtain such an improvement, we use a win-win approach, similar to what was suggested by Austrin et al.~\cite{DBLP:conf/stacs/AustrinKKN16}: if $p=O^*(2^{0.72n})$ is a valid upper bound for the sum of the squared frequencies of the entries in the lists, Theorem~\ref{thm:LD-ST} already gives the promised running time. 
On the other hand if $p=O^*(2^{0.72n})$ is not a valid upper bound, then we exploit the particular additive combinatorics of subset sums to show 
that a different polynomial space algorithm works well for this Subset Sum instance.
Note that if $p$ is large then a large number of subsets must have the same sum. This does not directly seem exploitable in general, let alone using only polynomial space.
However, recently Austrin et al.~\cite{DBLP:conf/stacs/AustrinKKN16} developed a technique to show  
that in this case the number of distinct sums $|w(2^{[n]})|:= |\{\sum_{e \in X} w_e : X \in 2^{[n]} \}|$ must be substantially small, and this technique is almost directly applicable here. 
Observe that this connection is completely non-trivial. For example, a-priori it is entirely possible that one particular subset sum occurs $2^{n-1}$ times (so half the subsets have the same sum), and the other $2^{n-1}$ subset sums are all distinct (and hence $|w(2^{[n]})| = 2^{n-1}+1$). Nevertheless, Austrin et al. show that such connections can be obtained by relating the problem to bounding `Uniquely Decodable Code Pairs' (see e.g.~\cite{DBLP:conf/isit/AustrinKKN16} and the references therein), a well-studied notion in information theory. In our setting, we show that if $p=O^*(2^{0.72n})$ is not a valid upper bound as required in Theorem~\ref{thm:LD-ST}, then $|w(2^{[n]})| \leq O^*(2^{0.86n})$ and we can run a $O^*(|w(2^{[n]})|)$ time polynomial space algorithm (see e.g.~\cite{DBLP:conf/iwpec/KaskiKN12,DBLP:conf/stacs/AustrinKKN16}).

To obtain the result for Knapsack (and for Binary Linear Programming) we use a reduction to Subset Sum previously presented by Nederlof et al.~\cite{DBLP:conf/mfcs/NederlofLZ12}. 

\paragraph{Techniques behind Theorem~\ref{thm:LD-ST}:}
At its core, the algorithm behind this theorem relies on space efficient cycle finding algorithms that, given access to a function $f: [n] \rightarrow [n]$, are able to find and sample pairs $i,j$ such that $f(i)=f(j)$ efficiently. To gain any non-trivial improvement over naive search with these algorithms, certain random-like properties of the function $f$ are required. Cycle finding algorithms are well-known in cryptography (a famous example is the Pollard's rho method for factorization) where these random-like properties are assumed to hold, but are less commonly used in the traditional worst case analysis of algorithms.

These cycle finding algorithms were recently applied in a closely related setting by Beame et al.~\cite{beame2013element}. To guarantee the random-like properties of $f$,
\cite{beame2013element} uses the idea of (essentially) `shuffling' $f$ using the assumed random access to the read-only tape with random bits to ensure $f$ has the required properties. They use this to give a substantially improved space-time trade-off for the related Element Distinctness problem, in which one is given a list $x \in [m]^n$ and needs to determine whether all values occur at most once.

We essentially follow the ideas of Beame et al.~\cite{beame2013element} with some minor modifications which are mostly in the analysis.
In our setting, we only care whether there are two indices $i$ and $j$ such that $x_i=y_j$ (and not about the cases where $x_i=x_j$ or $y_j=y_j$). 
To this end, we 
use a simple idea from cryptography (see e.g.~\cite{DBLP:journals/joc/OorschotW99}, and for a non-rigorous application to \emph{random} Subset Sum,~\cite{DBLP:conf/eurocrypt/BeckerCJ11}), which merges the two given lists into one list, and
show that the distribution of returned pairs of indices by the algorithm of Beame et al. is close to a uniform sample.
To allow for a space-time trade-off, we use the elegant extension of the cycle finding algorithm by Beame et al.~\cite{beame2013element}.

\subsection{Related Previous Work}

In this section we briefly outline related previous work, grouped by four of the research areas touched by our work.


\paragraph{Exact algorithms for NP-complete problems.} The goal of solving NP-complete problems exactly using as few resources as possible has witnessed considerable attention in recent years.
The Subset Sum problem and its variants have received a lot of attention in this area, but as mentioned above, nothing better than exhaustive search (for polynomial space)
or the classic results of \cite{DBLP:journals/jacm/HorowitzS74,DBLP:journals/siamcomp/SchroeppelS81} was known. Because of the lack of progress on both questions, several works investigated space-time trade-offs indicating how fast the problems can be solved in $2^{\alpha n}$ space for some $0< \alpha <1$ (see e.g.~\cite{DBLP:conf/icalp/AustrinKKM13,DBLP:conf/crypto/DinurDKS12}). 
A promising research line aimed at improving the running time of~\cite{DBLP:journals/jacm/HorowitzS74} in the worst case is found in~\cite{DBLP:conf/stacs/AustrinKKN15,DBLP:conf/stacs/AustrinKKN16}. One of their results is that instances satisfying $|w(2^{[n]})| \geq 2^{0.997n}$ can be solved in $O^*(2^{0.49991n})$ time, and a curious byproduct is that every instance can either be solved faster than~\cite{DBLP:journals/jacm/HorowitzS74} or faster than the trivial polynomial space algorithm.

Exact algorithms for the Binary Linear Programming were previously studied by Impagliazzo et al.~\cite{impagliazzobip}, and Williams~\cite{DBLP:conf/stoc/Williams14}. Using the notation from Corollary~\ref{thm:bip}, the algorithm of Impagliazzo runs in time $O^*(2^{n(1-\poly(1/c))})$ if $d=cn$, and the algorithm of Williams runs in time $O^*(2^{n(1-\Omega(1/((\log \log m) (\log d)^5)))})$, and both algorithms use exponential space. Our algorithm only improves over brute-force search when $d = o(n/(\log(n \log(mn))))$, but improves over the previous algorithms in the sense that it uses only polynomial space.

\paragraph{Bounded Space Computation.}
In many applications, memory is considered to be a scarcer resource than computation time. Designing algorithms with limited space usage has led to extensive work in areas such as the study of complexity classes $L$ and $NL$ (see e.g.~\cite{DBLP:books/daglib/0023084}) and pseudorandom generators that fool bounded space algorithms (see e.g.~\cite{DBLP:journals/combinatorica/Nisan92}). In particular, its study elicited some surprising non-trivial algorithms for very fundamental computational problems such as graph reachability. In this paper we add the List Disjointness problem to the list.
Space-time trade-offs in many contexts have already been studied before (see e.g.~\cite[Chapter 7]{Levitin:2002:IDA:579301} or for trade-offs Exact algorithms for NP-complete problems,~\cite[Chapter 10]{DBLP:series/txtcs/FominK10}), and space-time lower bounds for the Element Distinctness problem were already obtained by e.g. Ajtai~\cite{ajtai}.

\paragraph{Cryptography.} In the area of cryptography the complexity of Subset Sum on random instances has received considerable attention \cite{DBLP:conf/stacs/FlaxmanP05,DBLP:journals/jacm/LagariasO85,Coster92improvedlow-density}, motivated by cryptographic schemes relying on the hardness of Subset Sum~\cite{DBLP:journals/joc/ImpagliazzoN96,Odlyzko90therise}. In a recent breakthrough Howgrave-Graham and Joux~\cite{DBLP:conf/eurocrypt/Howgrave-GrahamJ10} showed that random instances can be solved faster than the Meet-in-the-Middle approach, specifically in $O^*(2^{0.3113n})$ time. This was subsequently improved in~\cite{DBLP:conf/eurocrypt/BeckerCJ11} where also a $O^*(2^{0.72n})$ time \emph{polynomial space} algorithm was claimed\footnote{While this algorithm served as an inspiration for the current work and the claim looks reasonable, we felt several arguments were missing in the write-up that seem to rely on implicit strong assumptions of sufficient randomness.}.

%

\paragraph{$k$-Sum and List Disjointness.}
These are among the most fundamental algorithmic problems that arise as sub-routines in various problems. Most relevant is the paper by Beame et al.~\cite{beame2013element} which presents an algorithm for the Element Distinctness problem in which one is given a list $x \in [m]^n$ and needs to determine whether all values occur at most once, and our work heavily builds on this algorithm. Assuming the exponential time hypothesis, $k$-Sum requires $n^{o(k)}$ time~\cite{DBLP:conf/soda/PatrascuW10}, and while the problem can be solved in $n^{\lceil k/2 \rceil}$ via Meet-in-the-Middle it is sometimes even conjectured that $k$-Sum requires $n^{k/2 -o(1)}$ time. Before our work, no non-trivial algorithms using only polylogarithmic space were known. 
As List Disjointness and $2$-Sum are (up to a translation of one of the sets) equivalent, our algorithm for List Disjointness also implies new space efficient algorithms for random instances $k$-Sum for $k \geq 2$. Space-time trade-offs for $k$-Sum were studied by Wang~\cite{DBLP:conf/esa/Wang14}, and more recently, by Lincoln et al.~\cite{DBLP:conf/icalp/LincolnWWW16}. Our results improve these trade-offs for sufficiently small space on random instances assuming random access to polynomially many random bits. 

We are not aware of previous work explicitly targeting random $k$-Sum, and would like to mention that it is currently not clear whether random $k$-Sum is easier than $k$-Sum in the sense that it can be solved faster given unbounded space. 

\subsection{Outline of This Paper}
In Section~\ref{sec:prel} we introduce the notation used in this paper. In Section~\ref{sec:ld} we present our List Disjointness algorithm. In Section~\ref{sec:sss} we present our algorithms for the Subset Sum, Knapsack and Binary Linear Programming problems. In Section~\ref{sec:ksum} we present our algorithms for $k$-Sum. Finally, in Section~\ref{sec:conc} we provide concluding remarks and possible directions for further research.

%% file: preliminaries.tex
\section{Notation}
\label{sec:prel}
For sets $A,B$ we denote $A^B$ as the set of all vectors with entries from $A$ indexed with elements from $B$. We freely interchange elements from $A^B$ with functions $f:B \rightarrow A$ in the natural way, and let $f^{-1}: A \rightarrow 2^{B}$ denote the inverse of $f$, i.e. for $a \in A$ we let $f^{-1}(a)$ denote $\{b \in B : f(b)=a\}$. If $f: B \rightarrow B$ and $i \geq 0$ we let $f^i(s)$ denote $s$ if $i=0$ and $f(f^{i-1}(s))$ otherwise. If $f: B \rightarrow A$ and $\mathcal{F} \subseteq 2^B$ we also denote $f(\mathcal{F})$ for $\{ f(X) : X \in \mathcal{F}\}$. If $f: B \rightarrow \mathbb{Z}$ and $X \subseteq B$ we shorthand $f(X)=\sum_{e \in X} f(e)$. 

For vectors $x,y \in \mathbb{Z}^{B}$ we denote by $\langle x,y\rangle$ the inner product of $x$ and $y$, i.e. $\langle x,y\rangle= \sum_{i \in B}x_iy_i$. For integers $p,q,r$ we let $q \equiv_p r$ denote that $q$ is equal to $r$ modulo $p$. If $A$ is a set, $a \in_R A$ denotes $a$ is uniformly at random sampled from the set $A$.
We let $\tilde{O}(T)$ denote expressions of the type $O(T\cdot \mathrm{polylog}(T))$, while $O^*()$ denotes expressions of the type $O(T\cdot \poly(|x|))$ where $x$ is a problem instance at hand. For integers $i\leq j$, $[i]$ denotes the set $\{1,\ldots,i\}$ and $[i,j]$ denotes $\{i,i+1,\ldots,j\}$.
If $G=(V,E)$ is a directed graph and $v \in V$, we let $N^-_G(v)$ denote the set of in-neighbors of $v$ and $N^+_G(v)$ denote the set of out-neighbors of $v$. The in-degree of $v$ is $|N^-_G(v)|$; the out-degree of $v$ is $|N^+_G(v)|$, and a graph is $k$-(in/out)-regular if all vertices have (in/out)-degree $k$. 

By Monte Carlo algorithms we mean randomized algorithms with only false negatives happening with constant probability. If the algorithms in this work claim an instance is a YES-instance they can also provide a witness. With random read-only access to a function we mean we can read of any of its evaluations in constant time.

%% file: listdisjointness.tex
\section{List Disjointness}\label{sec:ld}

In this section we prove Theorem~\ref{thm:LD-ST}, which we restate here for convenience.

\begin{reptheorem}{thm:LD-ST}
There is a Monte Carlo algorithm that takes as input an instance $x,y \in [m]^n$ of List Disjointness with $m \leq \poly(n)$, an integer $s$, and an upper bound $p$ on the sum of the squared frequencies of the entries in the lists, i.e. $\sum_{v=1}^m |x^{-1}(v)|^2+|y^{-1}(v)|^2 \leq p$, and solves the List Disjointness instance $x,y$ in time $\tilde{O}(n\sqrt{p/s})$ and $O(s \log n)$ space for any $s \leq n^2/p$. The algorithm assumes random read-only access to a random function $h:[m] \rightarrow [n]$.
\end{reptheorem}

Naturally, the set $[m]$ above can be safely replaced with any $m$-sized set (and indeed we will apply this theorem in a setting where values can also be negative integers), but for convenience we stick to $[m]$ in this section. We first recall and slightly modify the result from Beame et al.~\cite{beame2013element} in Subsection~\ref{subsec:beame} and afterwards present our algorithm and proof in Subsection~\ref{subsec:ldalg}.
For easier exposition and completeness, we also repeat some ideas and results from \cite{beame2013element}.

\subsection{The Collide Function from Beame et al.} \label{subsec:beame}
At its core, our algorithm for List Disjointness uses a space-efficient subroutine to find cycles in $1$-out-regular directed graphs i.e. graphs in which every vertex has out-degree exactly $1$.
If $G=(V,E)$ is such a graph, its edge relation can be captured by a function $f:V\rightarrow V$ where $f(u)=v$ if vertex $u$ has an outgoing edge to vertex $v$. We say that two vertices $u\neq v$ are \emph{colliding} if $f(u)=f(v)$ and refer to $(u,v)$ as a \emph{colliding pair}. 
For $k \in V$, denote by $f^*(k)$ the set of vertices reachable from $k$ in $G$, or more formally $f^*(k)=\{f^i(k): i \geq 0\}$ (recall from Section~\ref{sec:prel} that $f^{i}$ denotes the result of applying $f$ iteratively $i$ times to $k$). Since $V$ is finite, clearly there exists a unique smallest $i < j$ such that $f^i(k)=f^j(k)$. The classic and elegant Floyd's cycle finding algorithm takes $k \in V$ as input and finds such $i$ and $j$, using read-only access to $f$, $\tilde{O}(|f^*(k)|)$ time and $O(\log n)$ working memory, where $n$ denotes $|V|$.\footnote{Sometimes $O(1)$ memory usage is stated assuming the RAM model with word size $\log(|V|)$, but as $|V|$ may be exponential in the input size in our applications we avoid this assumption.}
See e.g. the textbook by Joux on cryptanalysis~\cite[Section 7.1.1.]{joux2009algorithmic} or Exercises~6 and~7 of Knuth's `Art of Computer Programming'~\cite{knuth1981seminumerical} for more information.

An extension of this algorithm allowing space-time trade-offs was recently proposed by Beame et al.~\cite{beame2013element} to find multiple colliding pairs in $G$ more efficiently than using independent runs of Floyd's algorithm. Formally, given a sequence $K=(k_1,\dots,k_s)$ of $s$ vertices in $G$, define $f^*(K):=\cup_{i=1}^s f^*(k_i)$ to be the set of vertices reachable in $G$ from $K$. Beame et al. show there is a deterministic algorithm that finds all colliding pairs in the graph $G[f^*(K)]$ in time $O(|f^*(K)|\log s\min\{s,\log n\})$ using only $O(s\log n)$ space. 
We will apply this but abort a run after a certain time in order to facilitate our analysis, and use the following definition to describe its behavior.

\begin{definition}
	Given a sequence $K = (k_1,\ldots,k_s)$ and an integer $L$, let $\ell$ be the greatest integer upper bounded by $s$ satisfying $|f^*(\{k_1,\ldots,k_\ell\})|\leq L$. Define $f^*_L(K)=f^*(\{k_1,\ldots,k_\ell\})$.
\end{definition}

\begin{lemma}[Implicit in~\cite{beame2013element}]\label{lem:beameetal}
	There is a deterministic algorithm $\mathtt{Collide}(f, K, L)$ that given read-only access to $f: V \rightarrow V$ describing a $1$-out-regular directed graph $G$ and a set $K$ of $s$ vertices
	  returns the set of pairs $\{(v,N^{-}_G(v) \cap f^*_L(K)): |N^{-}_G(v) \cap f^*_L(K)|>1\}$ using $O(L\log s\min\{s,\log n\})$ time and $O(s \log n)$ space.
\end{lemma}

The algorithm of Beame et al.~\cite{beame2013element} can be thought of as follows. First pick the vertex $k_1$ and follow the path dictated by $f$, i.e. $f^1(k_1),f^2(k_1),\ldots,$ until a colliding pair is found (or we revisit the start vertex $k_1$). Thus we stop at the first step $t$ such that there exists a $t'<t$ with $f^{t'}(k_1)=f^{t}(k_1)$. This colliding pair of vertices is reported (if we not revisit the start vertex) and then a similar process is initiated from the next vertex, $k_2$, in $K$ until a colliding pair is found again. Notice that this time a colliding pair might occur because of two reasons: first because the path from $k_2$ intersects itself, and second because the path from $k_2$ intersects the path of $k_1$. We again report the found colliding pair and move on to the next vertex in $K$ and so on. To achieve the aforementioned time and space bounds, Beame et al.~\cite{beame2013element} use Floyd's cycle finding algorithm combined with additional bookkeeping. By inspecting their algorithm it is easily seen that when aborting after $O(L\log s\min\{s,\log n\})$ time, all colliding pairs in $f^*_L(K)$ will be reported.
Note that the output of $\mathtt{Collide}(f, K, L)$ can indeed be described in $O(s \log n)$ space as each vertex in $K$ either gives rise to a new $(v,X)$ pair in the output or adds one vertex to the set $X$ for some pair $(v,X)$.

\subsection{The Algorithm for List Disjointness}\label{subsec:ldalg}

Given a list $z\in [m]^n$, define the number of \textit{pseudo-solutions} of $z$ as 
\[p(z)=\sum_{v=1}^m |z^{-1}(v)|^2 .\]
This is a measure of how many pairs of indices (i.e. `pseudo-solutions') there are in which $z$ has the same value. Note that a similar notion of `pseudo-duplicates' was used by Beame et al.~\cite{beame2013element}. Notice that $p(z) \geq n$ as $ p(z)\ge \sum_v  |z^{-1}(v)|=n$. The running time of our algorithm to find a common value of two lists $x,y$ will depend on the quantity
\[
	p(x,y)=p(x)+p(y).
\]
We refer to a \emph{solution} as a pair of indices $i^*,j^*$ such that $x_{i^*}=y_{j^*}$. On a high level, our algorithm uses the procedure $\mathtt{Collide}$ to obtain samples from the union of the set of solutions and pseudo-solutions of $x$ and $y$. The quantity $p(x,y)$ is thus highly relevant as it indicates the number of samples required to guarantee that a solution will be found with good probability. We first present an algorithm for the following promise version of the List Disjointness problem, from which Theorem~\ref{thm:LD-ST} is an easy consequence.

\begin{lemma}\label{lem:LD-ST}
There is a randomized algorithm that given two lists $x,y \in [m]^n$ with $m \leq \poly(n)$ and a common value given by $x_{i^*}=y_{j^*}$, and positive integers $p(x,y) \leq p$ and $s$ satisfying $s \le n^2/p$, outputs the solution $(i^*,j^*)$ with constant probability using $\tilde{O}\left(n\sqrt{p/s}\right)$ expected time and $O(s \log n)$ space. The algorithm assumes random read-only access to a random function $h:[m] \rightarrow [n]$.
\end{lemma}
Note that the algorithm from Lemma~\ref{lem:LD-ST} can be easily turned into a decision algorithm for List Disjointness: run the algorithm provided by Lemma~\ref{lem:LD-ST}, and return ``disjoint" if it is still running after $\tilde{O}(n\sqrt{p/s})$ time. If a solution exists, this algorithm will find a common element with constant probability by Markov's inequality. Thus this lemma implies Theorem~\ref{thm:LD-ST}, and the remainder of this section is devoted to the proof of Lemma~\ref{lem:LD-ST}. To this end, we assume indices $i^*$ and $j^*$ with $x_{i^*}=y_{j^*}$ exist. The algorithm is listed in Figure~\ref{alg:LD-ST}.
\newcommand{\bin}{\mathsf{bin}}
\begin{figure}[H]
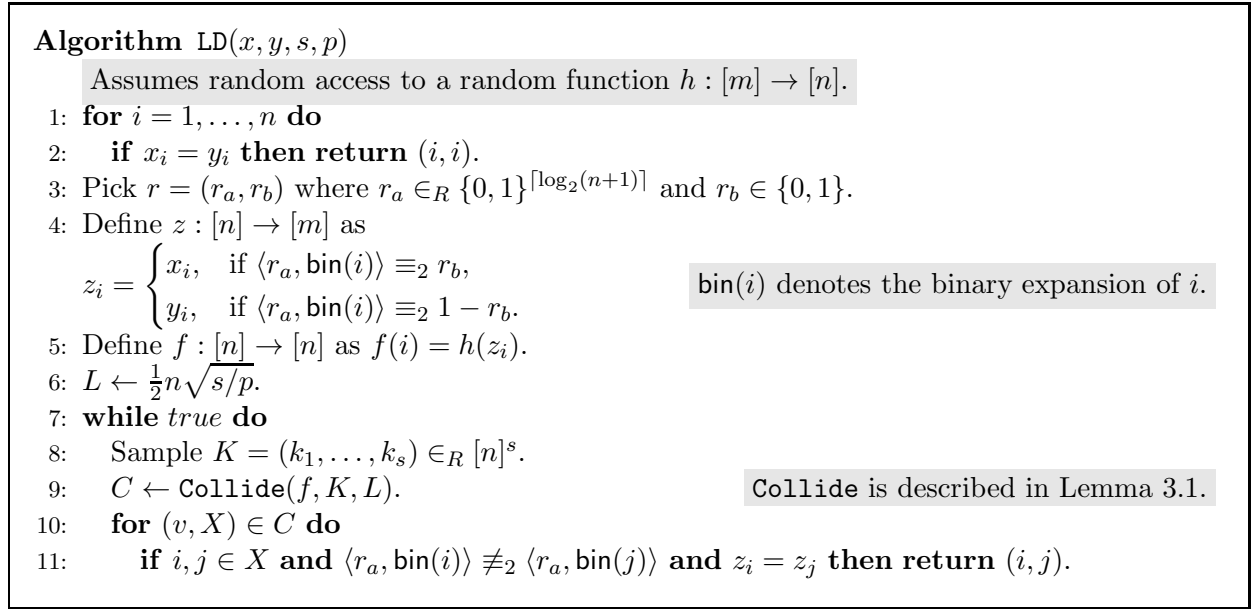

\begin{framed}
\begin{algorithmic}[1]
\REQUIRE $\mathtt{LD}(x,y,s,p)$
\item[]\algorithmiccommentt{Assumes random access to a random function $h: [m]\rightarrow [n]$.}
\FOR{$i=1,\ldots,n$}\label{lin:check}
	\LineIf{$x_i=y_i$}{\algorithmicreturn\ $(i,i)$.}\label{lin:trivial2}
\ENDFOR
	\STATE Pick $r=(r_a,r_b)$ where $r_a\in_R \{0,1\}^{\lceil \log_2 (n+1) \rceil}$ and $r_b \in \{0,1\}$.
	\STATE Define  $z : [n] \to [m]$ as
	\item[] $z_i =
	\begin{cases}
		x_i, & \mbox{if } \langle r_a,\bin(i) \rangle \equiv_2 r_b,\\
		y_i, & \mbox{if } \langle r_a,\bin(i) \rangle \equiv_2 1-r_b.
	\end{cases}$\hfill \algorithmiccommentt{$\mathsf{bin}(i)$ denotes the binary expansion of $i$.}
	\STATE Define $f : [n] \to [n]$ as $f(i) = h(z_i)$.
	\STATE $L \gets \tfrac{1}{2}n\sqrt{s/p}$.\label{lin:setL}
\WHILE{$true$} \label{lin:loop2}
	\STATE Sample $K = (k_1,\ldots,k_s) \in_R [n]^s$. \label{lin:sample}
	\STATE $C \gets \mathtt{Collide}(f, K, L)$. \hfill\algorithmiccommentt{$\mathtt{Collide}$ is described in Lemma~\ref{lem:beameetal}.} \label{lin:cycfind2}
	\FOR{$(v, X) \in C$}\label{lin:findc1}
		\LineIf{$i,j \in X$ \algorithmicand\ $\langle r_a, \bin(i) \rangle \not\equiv_2 \langle r_a,\bin(j) \rangle$ \algorithmicand\ $z_i = z_j$}{\algorithmicreturn~$(i,j)$.}\label{lin:findc2}
	\ENDFOR
\ENDWHILE
\end{algorithmic}
\end{framed}
\vspace{-1.5em}
\caption{Applying the collision search technique for list disjointness.}
\label{alg:LD-ST}
\end{figure}
The intuition behind this algorithm is as follows: at Line~\ref{lin:findc2} we have that $i\neq j$ forms a colliding pair, i.e. $f(i)=f(j)=v$, and since $h$ is a random hash function we expect $z_i=z_j$. This implies that $(i,j)$ is either a pseudo-solution or a solution which is checked at Line~\ref{lin:findc2} by the conditions $\langle r_a, \bin(i) \rangle \not\equiv_2 \langle r_a,\bin(j) \rangle$ and $z_i = z_j$. 

Recall that $f$ is naturally seen as a directed graph with vertex set $[n]$ (see Figure~\ref{fig:example} for an example), and therefore we also refer to the indices from $[n]$ as vertices. The main part of the proof of Lemma~\ref{lem:LD-ST} is to lower bound the probability that the solution $(i^*,j^*)$ is found in a single iteration of the loop of Line~\ref{lin:loop2}. Let us consider such an iteration, and use $W_L(K)$ (mnemonic for `walk') to denote the length $L$ prefix of the following sequence of vertices 
\begin{equation*}
	(k_1,f(k_1),\dots,f^{\ell_1}(k_1),k_2,f(k_2),\dots,f^{\ell_2}(k_2),\ldots,k_s,f(k_s),\ldots,f^{\ell_s}(k_s)),
\end{equation*}
where $\ell_i$ is the smallest positive integer such that $z(f^{\ell_i}(k_i))=z(f^{p}(k_j))$ and either $j=i$ and $p<\ell_i$ or $j<i$ and $p< \ell_j$. More intuitively stated, we continue with a next element of $K$ in this sequence whenever we reach an index $i$ for which we already encountered the value $z(i)$. Notice that $f^{\ell_i}(k_i)$ might either have the same $z$ value as a previous \textit{distinct} vertex which means that these two vertices form a solution or a pseudo-solution, or is a repeated vertex which happens due to the hash function $h$ mapping different $z$ values to the same vertex.

Suppose that both $i^*,j^*$ occur in $W_L(K)$, $\langle r_a, \bin(i^*) \rangle\equiv_2r_b$ and $\langle r_a, \bin(j^*) \rangle\equiv_21-r_b$. Then $i^*,j^* \in f^*_L(K)$ as all vertices in $W_L(K)$ not in $f^*_L(K)$ are exactly the vertices of the path $k_\ell,f^1(k_\ell),\ldots$ that was aborted before it intersects either itself or other visited vertices, and thus have unseen values.
Thus we have that $i^*,j^* \in X$, where $(f(i^*),X) \in C$ in Line~\ref{lin:findc1}, and the algorithm will return the solution $(i^*,j^*)$. Because of this, we now lower bound the probability of the event that both $i^*,j^*$ occur in $W_L(K)$. Call $r=(r_a,r_b)$ \emph{good} if it simultaneously holds that $\langle r_a, \bin(i^*) \rangle \equiv_2 r_b$ and $\langle r_a, \bin(j^*) \rangle\equiv_2 1-r_b$ i.e. $z_{i^*}=x_{i^*}$ and $z_{j^*}=y_{j^*}$. First we have the following easy observation.
\begin{observation}\label{obs:good}
If $i^* \neq j^*$, then $\Pr_r[r \text{ is good}] = 1/4$. \label{eqngood}
\end{observation}
Indeed, this follows as $\Pr_r[\langle r_a,\bin(i^*) \rangle\not\equiv_2 \langle r_a,\bin(j^*) \rangle] = 1/2$ (which is observed by deferring the random choice $r_a(j)$ on a bit position $j$ where $\bin(i^*)$ and $\bin(j^*)$ differ) and conditioned on these inner product being of different parities, exactly one choice of $r_b$ leads to $r$ being good.

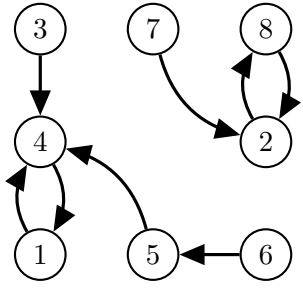
\begin{figure}[t]
	\centering
		\begin{tikzpicture}[scale=0.5]
		\tikzset{>=triangle 45}
		\begin{scope}[every node/.style={circle,thick,draw}]
		\node (1) at (0,0) {$1$};
		\node (2) at (6,3) {$2$};
		\node (3) at (0,6) {$3$};
		\node (4) at (0,3) {$4$};
		\node (5) at (3,0) {$5$};
		\node (6) at (6,0) {$6$} ;
		\node (7) at (3,6) {$7$} ;
		\node (8) at (6,6) {$8$} ;
		\end{scope}
		
		\path [very thick,->,{bend left = 30}] (1) edge  (4);
		\path [very thick,->,{bend left = 30}] (2) edge  (8);
		\path [very thick,->] (3) edge  (4);
		\path [very thick,->,{bend left = 30}] (4) edge  (1);
		\path [very thick,->,{bend right = 30}] (5) edge  (4);
		\path [very thick,->] (6) edge  (5);
		\path [very thick,->,{bend right = 30}] (7) edge  (2);
		\path [very thick,->,{bend left = 30}] (8) edge  (2);
		\end{tikzpicture}
		\caption{An example of the graph $G$ representing $f$ with $n=8$, $L=\infty$, $z=(11,7,3,8,3,4,1,1)$, $h(v)= (v\bmod 8) +1$ (so $f(i)=(z_i \bmod 8)+1$). If $K=(3,5,7)$ then $W_L(K)=(3,4,1,4,5,7,2,8)$ (note $4$ appears twice as $h$ maps both $11$ and $3$ to $4$) and $\mathtt{Collide}(f, K, L)$ outputs $\{(4,\{1,3,5\}), (2,\{7,8\})\}$.}
		\label{fig:example}
\end{figure}

%
%

We continue with lower bounding the probability that $i^*$ and $j^*$ occur in $W_L(K)$ given that $r$ is good. To do so, we start with an observation which allows us to replace the random string $W_L(K)$ by an easier to analyze random string. A similar connection was also used in~\cite[Proof of Theorem 2.1]{beame2013element}. We use $[n]^L$ to denote the set of all strings of length $L$ over the alphabet $\{1,2,\dots,n\}$. For any string $\beta$ in $[n]^L$, a \emph{repetition with respect to $z$} is said to occur at position $i$ of $\beta$ if there exists $j<i$ such that $z(\beta(i))=z(\beta(j))$. If clear from the context the `with respect to $z$' part is omitted. An $s^{\text{th}}$ repetition is said to occur at position $i$ of $\beta$ if already $s-1$ repetitions have occurred before index $i$ of $\beta$, and a repetition occurs at index $i$.

\begin{lemma}\label{obseqd}
	Fix any value of $s$ and $z$. Let $\beta$ be a random string generated by cutting off a uniform sample from the set $[n]^L$ at the $s^{\text{th}}$ repetition with respect to $z$. Then strings $W_L(K)$ and $\beta$ are equally distributed. I.e. for every $\rho \in [n]^{L'}$ with $L' \leq L$ we have $\Pr_{h,K}[W_L(K)=\rho]=\Pr_\beta[\beta=\rho]$.
\end{lemma}
\begin{proof}
	Note we may assume that the $s^{\text{th}}$ repetition (with respect to $z$) occurred at position $L'$ in $\rho$, or $L'=L$ and at most $s$ repetitions have occurred in $\rho$, as otherwise $\Pr_{h,K}[W_L(K)=\rho]=\Pr_\beta[\beta=\rho]=0$. Assuming $\rho$ has this property, $\Pr_\beta[\beta=\rho]=n^{-L'}$ as the only relevant event is that the first $L'$ locations of the uniform sample from $[n]^{L}$ used to construct $\beta$ match with $\rho$.
	
	For notational convenience, let us denote $\alpha=W_L(K)$. Note that
	\begin{equation}\label{eq:probchain1}
		\Pr_{h,K}[\alpha=\rho] = \prod_{i=1}^{L'} \Pr_{h,K}[\alpha_i=\rho_i | (\alpha_1,\ldots,\alpha_{i-1})=(\rho_1,\ldots,\rho_{i-1})].
	\end{equation}
	Thus to finish the claim it suffices to show that, for every $i \leq L'$, $\alpha_i$ is uniform over $[n]$ given $(\alpha_1,\ldots,\alpha_{i-1})=(\rho_1,\ldots,\rho_{i-1})$. We distinguish two cases based on whether the value $z(\alpha_{i-1})$ equals $z(\alpha_{i'})$ for some $i' < i-1$. If so, it follows from the definition of $W_L(K)=\alpha$ that $i-1=\ell_j$ for some $j$ and $\alpha_i=k_{j+1}$ which is chosen uniformly at random from $[n]$ and independent from all previous outcomes. Otherwise, $\alpha_i=h(z(\alpha_{i-1}))$ and since $z(\alpha_{i-1})$ does not occur in the sequence $z(\alpha_1),\ldots,z(\alpha_{i-2})$ and $h$ is random, $h(z(\alpha_{i-1}))$ is uniform over $[n]$ and independent of all previous outcomes.   
	
	Thus, each term of the product in~\eqref{eq:probchain1} equals $1/n$ and $\Pr_{h,K}[\alpha=\rho]=\Pr_\beta[\beta=\rho]=n^{-L'}$.
\end{proof}

This implies that instead of analyzing the sequence $W_L(K)$ we can work with the above distribution of strings in $[n]^L$ which is easier to analyze, as we will do now:




\begin{lemma}\label{lem:lwbndprob}
Fix any value of $s$ and $z$ with $p(z)\leq p$. Let $\beta$ be a random string generated by cutting off a uniform sample from $[n]^L$ at the $s^\text{th}$ repetition. Then $\Pr_\beta[\beta \text{ contains }i^* \text{ and }j^*] \geq \Omega((L/n)^2)$.
\end{lemma}
\begin{proof}
Letting $\beta$ be distributed as in the lemma statement, note that

\begin{align}
 \begin{split}
  & \Pr_{\beta}[\beta \text{ contains }i^*, j^*] \geq \Pr_{\tilde{\beta} \in_R [n]^L}[\tilde{\beta} \text{ contains }i^*, j^* \wedge \tilde{\beta} \text{ has at most $s$ repetitions}],
 \end{split}
\end{align}
because if the uniform sample from $[n]^L$ used to construct $\beta$ contains $i^*,j^*$ and has at most $s$ repetitions, $\beta$ contains $i^*$ and $j^*$. Let $E$ be the event that $\tilde{\beta}$ contains $i^*,j^*$ exactly once. Then the latter displayed expression in the inequality above can be further lower bounded as
\begin{align}
	&\geq \Pr_{\tilde{\beta} \in_R [n]^L}[E \wedge \tilde{\beta} \text{ has at most $s$ repetitions}]\nonumber\\
	&= \Pr_{\tilde{\beta} \in_R [n]^L}[E] \label{eq:leftprob}\\
	&\cdot \Pr_{\tilde{\beta} \in_R [n]^L}[\tilde{\beta} \text{ has at most $s$ repetitions}| E]. 	\label{eq:rightprob}
\end{align}

We may lower bound~\eqref{eq:leftprob} as
\begin{equation}
\label{eq:contains}
\begin{aligned}
\Pr_{\tilde{\beta} \in_R [n]^L}[\tilde{\beta} \text{ contains } i^* \text{ and } j^* \text{ exactly once}] &= \Pr_{\tilde{\beta} \in_R [n]^L}[\exists! \ell_1,\ell_2 \le L: \tilde{\beta}(\ell_1) = i^* \wedge \tilde{\beta}(\ell_2) = j^*]\\
	&= L(L-1)(1/n)^2 (1-2/n)^{L-2}\\
	&\geq L(L-1)(1/n)^2(1-2/n)^{n} = \Omega((L/n)^2).
\end{aligned}
\end{equation}
The second equality uses that for each $v\in [n]$, $\Pr[v=\tilde{\beta}(i)]=1/n$ for each $i$, and the $\tilde{\beta}(i)$'s are independent. The second inequality uses that $L\le n$ which is implied by $s\le n^2/p$ and $p \ge n$.

We continue by lower bounding~\eqref{eq:rightprob}. To this end, let $\beta'$ be the string obtained after removing the unique $\ell_1$ and $\ell_2$ indices from $\tilde{\beta}$ which contain $i^*$ and $j^*$ as values respectively. Note that $\beta'$ is uniformly distributed over $([n] \setminus \{i^*,j^*\})^{L-2}$, since we condition on $i^*$ and $j^*$ being contained exactly once in $\tilde{\beta}$. Let $q$ and $q'$ be the number of repetitions in $\tilde{\beta}$ and $\beta'$ respectively. Note that $q-q'=1$ if $z(\beta'(k)) \neq z(i^*)$ for every $k \in [L-2]$ because of the repetition in $\tilde{\beta}$ occurring at $\ell_2$ (assuming $\ell_1 < \ell_2$ for simplicity). Otherwise, $q-q'=2$ because of the repetition in $\tilde{\beta}$ occurring at $\ell_2$ and the repetition in $\tilde{\beta}$ occurring at $\ell_1$ if $z(i^*)=z(\tilde{\beta}(k))$ for $k < \ell_1$ or otherwise at $k$ if $z(i^*)=z(\tilde{\beta}(k))$ for $k \geq \ell_1$. Therefore we have that
\begin{align*}
	\mathbb{E}_{\tilde{\beta} \in_R [n]^L}[q] &= 1+ \Pr_{\beta'}[\exists k \in [L-2]: z(\beta'(k))=z(i^*)]+ \mathbb{E}_{\beta'}[q'] &\\
	 &\leq 1+ (L-2)\frac{|z^{-1}(z(i^*))|}{n-2}+ \mathbb{E}_{\beta'}[q']  & \hfill \algorithmiccommentt{Using $|z^{-1}(z(i^*))| \leq \sqrt{p}$, $\frac{L-2}{n-2}\leq \frac{L}{n}$.}\\
	 &\leq 1+L\sqrt{p}/n+\mathbb{E}_{\beta'}[q'] \leq 1+s/2+\mathbb{E}_{\beta'}[q'],& \hfill \algorithmiccommentt{Using $L=\tfrac{1}{2}n\sqrt{s/p}$.}
\end{align*}
where $\beta' \in_R \{[n] \setminus \{i^*,j^*\} \}^L$.
%
Since the number of repetitions is always less than the number of collisions (i.e. unordered pairs of distinct indices in which the string has equal values), we have the upper bound 
\begin{align*}
	\mathbb{E}_{\beta'}[q'] &\leq \sum_{i<j} \Pr_{\beta'}[z(\beta'(i))=z(\beta'(j))]\\
   				   &\leq \frac{\binom{L}{2} \left( \sum_{v=1}^m |x^{-1}(v)|^2+|y^{-1}(v)|^2+|x^{-1}(v)||y^{-1}(v)| \right) }{ (n-2)^2} \\
				   &\leq \binom{L}{2} 2p / (n-2)^2   \leq s/3,
\end{align*}
where we use the AM-GM inequality in the second inequality to obtain the upper bound $|x^{-1}(v)||y^{-1}(v)|\le |x^{-1}(v)|^2+|y^{-1}(v)|^2$. In the last inequality we assume $n$ is sufficiently large and use $\binom{L}{2} \leq L/2$. Thus $\mathbb{E}_{\tilde{\beta}}[q] \leq 1+5s/6$, and we can use Markov's inequality to upper bound $\Pr[q > s]=\Pr[q \geq s+1]\leq \frac{1+5s/6}{s+1}\leq 11/12$. This lower bounds~\eqref{eq:rightprob} by $1/12$ and therefore the claim follows.
\end{proof}

\begin{proof}[Proof of Lemma~\ref{lem:LD-ST}]
The case $i^*=j^*$ is checked at Line~\ref{lin:trivial2}. So we may assume $i^*\neq j^*$. By Observation~\ref{obs:good}, $r$ is good with probability $1/4$. In the remainder of this proof, we fix such a good $r$ (and thus $z$ is also fixed as it is determined by $r$).

We first analyze the running time performance of a single iteration of the loop from Line~\ref{lin:loop2}. Line~\ref{lin:sample} takes $O(L)$ time and $O(s\log n)$ space as $s\le L$, which follows from the requirement of $s$ being at most $n^2/p$. For Line~\ref{lin:cycfind2}, recall that $\mathtt{Collide}(f,K,L)$ takes $\tilde{O}(L)$ time and $O(s\log n)$ space. Lines~\ref{lin:findc1} and~\ref{lin:findc2} in the algorithm take $\tilde{O}(L)$ time and $O(s\log n)$ space as in Line~$\ref{lin:findc2}$ we can iterate over all elements of $X$ and quickly check whether a new element is in a solution by storing all considered elements of $X$ in a lookup table. 
Thus a single iteration of the while loop from Line~\ref{lin:loop2} uses $\tilde{O}(L)$ time and $O(s\log n)$ space.

As Lines~\ref{lin:check} to~\ref{lin:setL} clearly do not form any bottleneck, it remains to upper bound the expected number of iterations of the while loop before a solution is found.

To this end, recall that a solution $(i^*,j^*)$ is found if $i^*$ and $j^*$ are in $W_L(K)$. The probability that this happens is equal to the probability that $i^*,j^*$ are in a sequence $\beta$ generated by cutting off a uniform sample from the set $[n]^L$ at the $s^{\text{th}}$ repetition by Lemma~\ref{obseqd}. Applying Lemma~\ref{lem:lwbndprob} we can lower bound this probability with $\Omega((L/n)^2)$. This implies that
\[
	\E_{r,h}[\Pr_{K}[\mathtt{Collide}(f, K, L) \text{ detects } (i^*,j^*)]] \geq \Omega((L/n)^2),
\]

and therefore the expected number of iterations of the loop at Line~\ref{lin:loop2} until $(i^*,j^*)$ is found is $O(n^2/L^2)$. Since every iteration takes $\tilde{O}(L)$ time, the expected running time will be $\tilde{O}(n^2/L)$ which is $\tilde{O}(n\sqrt{p/s})$.
\end{proof}

We would like to mention that in the proof of Lemma~\ref{lem:lwbndprob}, upper bounding the number of repetitions by the number of collisions seems rather crude, but we found that tightening this step did not lead to a considerable improvement of an easy $\tilde{O}(n^2/s)$ time $O(s \log n)$ space algorithm for an interesting value of $s$.

%% file: subsetsum.tex
\section{Subset Sum, Knapsack and Binary Linear Programming}\label{sec:sss}
The main goal of this section is to prove Theorem~\ref{thm:sss}, which we restate here for convenience.

\begin{reptheorem}{thm:sss}
	There are Monte Carlo algorithms solving Subset Sum and Knapsack using
	$O^*(2^{0.86n})$ time and polynomial space. The algorithms assume random read-only access to random bits.
\end{reptheorem}
We first focus on Subset Sum, and discuss Knapsack later. 
Let $w_1,\ldots,w_n$ be the integers in the Subset Sum instance. 
Throughout this section we assume that $n$ is even, by simply defining $w_{n+1}=0$ if $n$ is odd.
Let $w$ be the weight vector $(w_1,\ldots,w_n)$.

To prove the theorem for Subset Sum, we use a `win-win' approach based on the number of all possible distinct sums $|w(2^{[n]})|=|\{\langle w,x \rangle: x\in \{0,1\}^n\}|$ generated by $w$. Specifically, if $|w(2^{[n]})|$ is sufficiently large, we prove that the number of pairs of subsets with the same sum, or more formally
 \[|\{ (x,y) \in  \{0,1\}^n \times \{0,1\}^n : \langle w,x \rangle = \langle w,y \rangle \}|, \]
cannot be too large (see Lemma~\ref{lem:sumsvscols} below for a more precise statement).
This is done by showing a smoothness property of the distribution of subset sums. The proof of this smoothness property builds for a large part on ideas by Austrin et al.~\cite{DBLP:conf/stacs/AustrinKKN16} in the context of exponential space algorithms for Subset Sum. Then we use the Meet-in-the-Middle approach with a random split $L,R$ of $[n]$, to reduce the Subset Sum instance to an instance of List Disjointness with $2^{n/2}$-dimensional vectors $x$ and $y$ whose entries are the sums of the integers from $w$ indexed by subsets $L$ and $R$, respectively. We apply Theorem~\ref{thm:LD-ST} to solve this instance of List Disjointness. The crux is that, assuming $|w(2^{[n]})|$ to be large, the smoothness property implies we may set the parameter $p$ sufficiently small. On the other hand, if $|w(2^{[n]})|$ is sufficiently small, known techniques can be applied to solve the instance in $O^*(|w(2^{[n]})|)$ time. 

The result for Knapsack and Binary Integer Programming is subsequently obtained via a reduction by Nederlof et al.~\cite{DBLP:conf/mfcs/NederlofLZ12}. 

\paragraph{On the smoothness of the distribution of subset sums.}
The required smoothness property is a quite direct consequence from the following more general bound independent of the subset sum setting. Let us stress here that the complete proof of the smoothness property proof is quite similar to~\cite[Proposition 4.4]{DBLP:conf/isit/AustrinKKN16}, which in turn was inspired by previous work on bounding sizes of Uniquely Decodable Code Pairs (e.g.~\cite{tilborg1}). However, our presentation will be entirely self-contained.
\begin{lemma}\label{lem:dsvscols}
	Let $d \leq n$ be a positive integer. Let $A \subseteq \{0,1\}^n$ and $B \subseteq \{-1,0,1\}^n$ be collections satisfying:
	(i) $|\mathrm{supp}(b)|=d$ for every $b \in B$, i.e.~any $b\in B$ has $d$ non-zero entries and,
	(ii) for every $a,a' \in A$ and $b,b' \in B$ the following holds:
	\begin{equation}\label{eq:udcp}
	a+b=a'+b' \text{ implies that } (a,b)=(a',b').
	\end{equation}
	Then $|A||B| \leq 2^n\binom{n}{\lceil d/2 \rceil}\poly(n)$.
\end{lemma}
\begin{proof}
	Note we may assume $d$ and $n$ are even since we can otherwise increase $n$ to $n+1$ and set $a_{n+1} = b_{n+1}=1$ (if we want to increase $d$) or $a_{n+1} = b_{n+1}=0$ (if we want to increase $n$) for every $b \in B$ and $a \in A$. The increases of $n$ and $d$ are subsumed by the $\poly(n)$ factor.
	
	The idea of the proof will be to give a short encoding of the pairs $(a,b) \in A \times B$ using the above properties.
	We use the standard sumset notations $A + B = \{a+b: a\in A, b\in B\}$ and $A-B=\{a-b: a \in A, b\in B\}$, where the additions of the vectors are in $\mathbb{Z}^n$. Note that~\eqref{eq:udcp} implies that $|A+B|=|A||B|=|A-B|$ as $a-b=a'-b'$ implies $a+b'=a'+b$.
	
	For a vector $a\in A$, recall that $a^{-1}(0)$ denotes the set of indices $i$ in $[n]$ with $a(i)=0$. Similarly, define the sets $a^{-1}(1)$ and $b^{-1}(u)$ for $u \in \{-1,0,1\}$. 
	For a pair $(a,b) \in A \times B$, let us define its signature $x(a,b)$ as the vector 
		\[
	x(a,b) = x =(x_{0,-1},x_{0,0},x_{0,1},x_{1,-1},x_{1,0},x_{1,1}) \in \mathbb{N}^6,
	\]
	where $x_{u,v} = |a^{-1}(u) \cap b^{-1}(v)|$ for each $u \in \{0,1\}$ and each $v \in \{-1,0,1\}$.
	Note that $x$ can take at most $(n+1)^6$ possible values.
	
	For a fixed signature vector $x$, define $P_x$ as the set of all pairs $(a,b) \in A \times B$ with signature $x(a,b)=x$. 
So, if $(a,b)\in P_x$, $x$ indicates for each combination of possible values the number of indices in which this combination occurs in the pair $(a,b)$.
For a vector $c \in \mathbb{Z}^n$, we let $ODD(c) \subseteq [n]$ denote the set of indices $i$ such that $c_i$ is odd. Letting $x_{odd}=x_{0,-1}+x_{0,1}+x_{1,0}$, note that for $(a,b) \in P_x$, we thus have that $|ODD(a+b)|=x_{odd}$.
 
To bound the number of pairs $(a,b)$, we will bound $P_x$ for each $x$. 
As $(a,b) \in A \times B$ is determined by $a+b$, it suffices to bound the number of sums $a+b$. Fix some $(a,b) \in P_x$. A crucial observation is that
 $a+b$ has entries from
$\{-1,0,1,2\}$, and a $2$ precisely occurs in the indices $a^{-1}(1) \cap b^{-1}(1)$ and a $-1$ precisely occurs at the indices in $a^{-1}(0) \cap b^{-1}(-1)$. So $a+b$  can be completely determined by specifying the triple 
    \[
	    (ODD(a+b),\ a^{-1}(0)\cap b^{-1}(-1),\ a^{-1}(1)\cap b^{-1}(1)).
	\]
Thus, we may bound $P_x$ by counting the number of possibilities of such triples to obtain
	\begin{equation}\label{eq:bnd1}
	\begin{aligned}
	|P_x| &\leq \binom{n}{x_{odd}}\binom{x_{odd}}{x_{0,-1}}\binom{n-x_{odd}}{x_{1,1}} \\ 
	& = \frac{n!}{ (x_{0,-1})! (x_{odd}-x_{0,-1})! (x_{1,1})! (n-x_{odd}-x_{1,1})!} \\
	& =\binom{n}{x_{0,-1},x_{odd}-x_{0,-1},x_{1,1},n-x_{odd}-x_{1,1}}.
	\end{aligned}
	\end{equation}
	Similarly, since $(a,b) \in A \times B$ is also determined by $a-b$, which is in turn determined by the triple
	\[
		(ODD(a+b),a^{-1}(0)\cap b^{-1}(1),a^{-1}(1)\cap b^{-1}(-1)),
	\]
	we may also bound $P_x$ by counting the number of possibilities of these triples to obtain
	\begin{equation}\label{eq:bnd2}
	\begin{aligned}
	|P_x| &\leq \binom{n}{x_{odd}}\binom{x_{odd}}{x_{0,1}}\binom{n-x_{odd}}{x_{1,-1}}\\
	&=\binom{n}{x_{0,1},x_{odd}-x_{0,1},x_{1,-1},n-x_{odd}-x_{1,-1}}.
	\end{aligned}
	\end{equation}
	Note that~\eqref{eq:bnd1} and~\eqref{eq:bnd2} are equivalent modulo interchanging $x_{0,-1}$ with $x_{0,1}$ and $x_{1,1}$ with $x_{1,-1}$ (which is natural as $A-B= A+(-B)$).
	We know that $x_{0,-1}+x_{1,1}+x_{0,1}+x_{1,-1}=d$, and therefore either $x_{0,-1}+x_{1,1}\leq d/2$ or $x_{0,1}+x_{1,-1}\leq d/2$. Using~\eqref{eq:bnd1} in the first case or~\eqref{eq:bnd2} in the second case we obtain
	\begin{align*}
	|P_x| &\leq  \max_{d_1+d_2 \leq d/2}  \binom{n}{d_1,x_{odd}-d_1,d_2,n-x_{odd}-d_2}  \\
	&=  \max_{d' \leq d/2} \max_{d_1+d_2 = d'}  \binom{n}{d_1,x_{odd}-d_1,d_2,n-x_{odd}-d_2}
		\end{align*}
		For any $u,v$ with $u+v=s$, the term  $u! v!$ is minimized when $u=\lfloor s/2 \rfloor$ and $v=\lceil s/2 \rceil$. Applying this to the term above, once with $u=d_1$ and $v=d_2$ and once with 
		$u=x_{odd}-d_1$ and $v=n-x_{odd}-d_2$, we obtain that 
	\begin{align*}
	& |P_x| \leq   \max_{d' \leq d/2}  \binom{n}{\lceil d'/2 \rceil,\lfloor (n-d')/2 \rfloor, \lfloor d'/2 \rfloor, \lceil (n-d')/2 \rceil} \\
	&\leq  \max_{d' \leq d/2}  2^n \binom{n/2}{\lceil d'/2 \rceil}\binom{n/2}{\lfloor d'/2 \rfloor}  \\
	&\leq  \max_{d' \leq d/2}  2^n \binom{n}{d'} \leq 2^n \binom{n}{d/2}.
	\end{align*}
	The last step uses that $\binom{n}{u}\binom{n}{v} \leq \binom{2n}{u+v}$ for any $u,v$ and that $d/2 \leq n/2$ and hence the maximum is attained at $d'=d/2$.
	The lemma now follows directly from this bound as $|A||B|\leq \sum_{x}|P_x|\leq 2^n \binom{n}{d/2} (n+1)^6$.
\end{proof}
Now we use Lemma~\ref{lem:dsvscols} to obtain the promised smoothness property of the distribution of subset sums.
\begin{lemma}\label{lem:sumsvscols}
Let $w=(w_1,\ldots,w_n)$ be integers and $d \leq n$ be a positive integer. Denote $C_d= \{ x\in \{-1,0,1\}^n: \langle w,x\rangle =0 \wedge |\mathrm{supp}(x)|=d \}$. Then
\[
	|w(2^{[n]})|\cdot |C_d| \leq 2^n \binom{n}{d/2}\poly(n).
\]
\end{lemma}
\begin{proof}
Let $A \in \{0,1\}^n$ be a set of vectors such that $\langle w,a \rangle = \langle w,a' \rangle$ implies $a=a'$ for every $a,a' \in A$.
Note that such an $A$ satisfying $|A|=|w(2^{[n]})|$ can be found by picking one representative from the set $\{x \in \{0,1\}^n : \langle w,x \rangle = s \}$ for each $s \in w(2^{[n]})$. We apply Lemma~\ref{lem:dsvscols} with $B=C_d$. It remains to show that this pair $(A,B)$ satisfies~\eqref{eq:udcp}. To this end, note that if $a+b=a'+b'$, then 
\[
	\langle a,w \rangle = \langle a,w \rangle + \langle b,w \rangle=\langle a+b,w \rangle =	\langle a'+b',w \rangle=\langle a',w \rangle + \langle b',w \rangle=\langle a',w \rangle,
\]
using linearity of inner product and $\langle b,w \rangle=\langle b',w \rangle=0$ by definition of $B$. Therefore, $a=a'$ by definition of $A$ and since $a+b=a'+b'$ it follows that $b=b'$. 
\end{proof}

Now we are fully equipped to prove the first part of the main theorem of this section:

\begin{proof}[Proof of Theorem~\ref{thm:sss}(a)]
The algorithm is as follows

\begin{figure}[H]
\begin{framed}
\begin{algorithmic}[1]
\REQUIRE $\mathsf{SSS}(w,t,h)$
 \item[]\algorithmiccommentt{Assumes random access to a random function $h: [\sum_{i=1}^n w_i]\rightarrow \{0,1\}^{n/2}$.}
 	\STATE Run a polynomial space, $O^*(2^{0.86n})$ time algorithm that assumes $|w(2^{[n]})|\leq O^*(2^{0.86n})$ from e.g.~\cite[Theorem 1(a)]{DBLP:conf/iwpec/KaskiKN12} or~\cite{DBLP:conf/stacs/AustrinKKN16}.\label{lin:dft}
    \STATE Let $(L,R)$ be random partition of $[n]$ with $|L|=n/2$ and $|R|=n/2$.
		\label{lin:split}
	\STATE Let $x$ be the list  $(\sum_{e\in X}w_e)_{X \subseteq L}$ of length $2^{n/2}$.
	\STATE Let $y$ be the list  $(t-\sum_{e\in Y}w_e)_{Y \subseteq R}$ of length $2^{n/2}$.
	\STATE Run $\mathsf{LD}(x,y,1,O^*(2^{0.72n}))$ and cut off the running time after $O^*(2^{0.86n})$ time if a solution was still not found. \label{lin:LD}
	\LineIfElse{so far no solution was found}{\algorithmicreturn\ NO}{\algorithmicreturn\ YES}
\end{algorithmic}
\end{framed}
\vspace{-1.5em}
\caption{Reducing SSS to List Disjointness.}
\label{alg:SSS}
\end{figure}

Here, the algorithm for Line~\ref{lin:dft} is implemented by hashing all integers of the Subset Sum instance modulo a prime of order $O^*(2^{0.86n})$ and running the algorithm by Lokshtanov and Nederlof~\cite{DBLP:conf/stoc/LokshtanovN10} (as already suggested and used in~\cite{DBLP:conf/iwpec/KaskiKN12, DBLP:conf/stacs/AustrinKKN16}). By checking whether a correct solution has been found by self reduction and returning NO if not, we can assume that the algorithm has no false positives and false negatives with constant probability assuming $|w(2^{[n]})|\leq O^*(2^{0.86n})$.

We continue by analyzing this algorithm. First, it is clear that this algorithm runs in $O^*(2^{0.86n})$ time. For correctness, note that the algorithm never returns false positives as the algorithms invoked on Lines~\ref{lin:dft} and~\ref{lin:LD} also have this property. Thus it remains to upper bound the probability of false negatives. Suppose a solution exists. If $|w(2^{[n]})|\leq O^*(2^{0.86n})$, Line~\ref{lin:dft} finds a solution with constant probability, so suppose this is not the case. As in Lemma~\ref{lem:sumsvscols}, denote
\[
	C_d = \{ x\in \{-1,0,1\}^n: \langle w,x\rangle =0 \wedge |\mathrm{supp}(x)|=d  \}.
\]
Then we know by Lemma~\ref{lem:sumsvscols} that $2^{0.86n}|C_d| \leq O^*(2^{n}\binom{n}{\lceil d/2 \rceil }$), so $C_d \leq O^*(2^{0.14n}\binom{n}{\lceil d/2 \rceil })$. Let $w_L,w_R$ denote the restrictions of the vector $w$ to all indices from $L$ and $R$ respectively. Let us further denote
\begin{align*}
	P^L &=  \{(x,y) \in \{0,1\}^{L} \times \{0,1\}^{L}: \langle w_L, x \rangle = \langle w_L, y \rangle  \},\\
	P^R &=  \{(x,y) \in \{0,1\}^{R} \times \{0,1\}^{R}: \langle w_R, x \rangle = \langle w_R, y \rangle  \},\\
	p &= \max \{|P^L|,|P^R|\}.
\end{align*}
As $|P^L|$ and $|P^R|$ are exactly the number of pseudo-solutions of the List Disjointness instance $(x,y)$, it remains to show that $p\leq O^*(2^{0.72n})$ with constant probability. In particular, we will show that $|P^L| \leq O^*(2^{0.72n})$ with probability at least $3/4$. As $|P^L|$ and $|P^R|$ are identically distributed a union bound shows that $p \leq O^*(2^{0.72n})$ with probability at least $1/2$. This suffices for our purposes as then the running time of the List Disjointness algorithm is $O^*(2^{n/2}\sqrt{p})=O^*(2^{0.86n})$.
 
Note that elements of $C_d$ may contribute to $P^L$, but only if their support is not split, e.g. the support is a subset of either $L$ or $R$. Let us introduce the following notation for elements of $C_d$ that contribute to $p$:
\[C^L_d = \{ x \in C_d: \mathrm{supp}(x) \subseteq L \}.\]
Suppose that $(x,y) \in P^L$. Then $\langle w_L,x-y \rangle = 0$ and therefore $x-y \in C^L_d$. Now suppose $|\mathrm{supp}(x-y)|=d$. Then for a fixed vector $v=x-y$, there are only $2^{n/2-d}$ choices for the pair $(x,y)$ such that $v=x-y$, as for each index $i \in L$ with $v_i=1$ we have $x_i=1$ and $y_i=0$, for each index $i \in L$ with $v_i=-1$ we have $x_i=0$ and $y_i=1$, and for each $i \in L$ with $v_i=0$, either $x_i=y_i=0$ or $x_i=y_i=1$. Therefore we have the upper bound
\begin{equation}\label{eqpl1}
	|P^L| \leq \sum_{d=0}^{n/2} |C^L_d| 2^{n/2-d}.
\end{equation}
If $v \in C_d$, for a random split $(L,R)$ as picked in Line~\ref{lin:split} we see that
\begin{equation}\label{eqpl2}
	\Pr[v \in C^L_d] = \Pr[\mathrm{supp}(v) \subseteq L] = \binom{n-d}{n/2-d} / \binom{n}{n/2}. 
\end{equation}
Now we can combine all the work to bound the expectation of $|P^L|$ over the random split as
\begin{align*}
	\mathbb{E}[|P^L|] &\leq \sum_{d=0}^{n/2} \mathbb{E}[C^L_d] 2^{n/2-d} & \hfill \algorithmiccommentt{Using~\eqref{eqpl1}}\\
&\leq \sum_{d=0}^{n/2} \sum_{v \in C_d} \Pr[v \in C^L_d] 2^{n/2-d} & \\
 &= \sum_{d=0}^{n/2} \left(|C_d| \binom{n-d}{n/2-d}/\binom{n}{n/2} \right) 2^{n/2-d} &\hfill \hfill \algorithmiccommentt{Using~\eqref{eqpl2}}\\
 &= O^*\left( \sum_{d=0}^{n/2} 2^{0.14n}\binom{n}{\lceil d/2 \rceil} \binom{n-d}{n/2-d} /\binom{n}{n/2}\right) 2^{n/2-d} & \hfill \algorithmiccommentt{By Lemma~\ref{lem:sumsvscols} and $|w(2^{[n]})| \geq 2^{0.86n}$} \\
  &=  O^*\left(2^{-0.36n}\sum_{d=0}^{n/2} \binom{n}{\lceil d/2 \rceil} \binom{n-d}{n/2-d}/ 2^{d} \right), & \hfill \algorithmiccommentt{Using $\binom{n}{n/2} \geq 2^{n}/n$}
 \end{align*}




Omitting polynomial terms, taking logs, rewriting using $\log_2 \binom{b}{a}=b\cdot h(a/b)$ where $h(q)= -q \log_2 q - (1 - q) \log_2 (1 - q)$ is the binary entropy function, and denoting $\delta = d/n$, this reduces to
\[
	 \left( -0.36 + \max_{0\leq \delta \leq 1/2} h(\delta/2)+ h\left(\frac{1/2-\delta}{1-\delta}\right)(1-\delta) -\delta \right)n.
\]
Note that here we are allowed to replace the summation by a max as we suppress factors polynomial in $n$. 
By a direct Mathematica computation this term is upper bounded by $0.72n$ (where the maximum is attained for $\delta \approx 0.0953$), which
 implies that $\mathbb{E}[|P^L|] \leq O^*(2^{0.72n})$ as required.
\end{proof}

\paragraph{Knapsack and Binary Linear Programming.}

Now Theorem~\ref{thm:sss}(b) follows from Theorem~\ref{thm:sss}(a) by the following reduction:

\begin{lemma}[\cite{DBLP:conf/mfcs/NederlofLZ12}, Theorem 2]\label{lem:redsss}
If there exists an algorithm that decides the Subset Sum problem in $O^*(t(n))$ time and $O^*(s(n))$ space then there exists an algorithm that decides the Knapsack in $O^*(t(n))$ time and $O^*(s(n))$ space.
\end{lemma}

We would like to remark that the techniques from Section~\ref{sec:ld} do not seem to be directly applicable to the knapsack problem, so the reduction of~\cite{DBLP:conf/mfcs/NederlofLZ12} seems necessary. 

Using the methods behind the proof of Lemma~\ref{lem:redsss}, we can also obtain an algorithm for the Binary Linear Programming problem. In particular, we use the following result:

\begin{lemma}\label{thm:reduce}
	Let $U$ be a set of cardinality $n$, let $\omega: U \rightarrow \{-N,\ldots,N\}$ be a weight function, and let $l<u$ be integers. Then there is a polynomial-time algorithm that returns a set of pairs $\Omega=\{(\omega_1,t_1),\ldots,(\omega_K,t_K)\}$ with $\omega_i: U \rightarrow \{-N,\ldots, N\}$ and integers $t_1,\ldots,t_K \in [-N,N]$ such that
	\begin{inparaenum}[(1)]
		\item $K$ is $O(n \lg(nN))$, and
		\item for every set $X \subseteq U$ it holds that $\omega(X) \in [l,u]$ if and only if there exists an index $i$ such that $\omega_i(X) = t_i$.
	\end{inparaenum}
\end{lemma}
This result is a small extension of Theorem 1 by Nederlof et al.~\cite{DBLP:conf/mfcs/NederlofLZ12} (assuming $u-l \leq nN$), but the same proof from~\cite{DBLP:conf/mfcs/NederlofLZ12} works for this extension (the polynomial time algorithm is a direct recursive algorithm using rounding of integers, refer to~\cite{DBLP:conf/mfcs/NederlofLZ12} for details). 

\begin{repcorollary}{thm:bip}
	There is a Monte Carlo algorithm solving Binary Integer Programming instances with maximum absolute integer value $m$ and $d$ constraints in polynomial space and time $O^*(2^{0.86n}(\log(mn)n)^{O(d)})$. The algorithm assumes random read-only access to random bits.
\end{repcorollary}
\begin{proof}
Using binary search, we can reduce the optimization variant to a decision variant that asks for a $x \in \{0,1\}^n$ such that $\langle a^j, x \rangle \in [l_j, u_j]$ for $j=1,\ldots,d+1$. Using Lemma~\ref{thm:reduce}, we can reduce this problem to $O(\log(m)n)$ instances on the same number of variables and constraints but with $l_1=u_1$. Using the same reduction on each single obtained instance for $j=2,\ldots,d$ we obtain $(O(\log(m)n))^d$ instances on $n$ variables and $m$ constraints with $l_j=u_j$ for every $j$, and these instances can be easily enumerated with linear delay and polynomial space. An instance satisfying $l_j=u_j$ for every $j$ can in turn be reduced to a Subset Sum problem on $n$ integers in a standard way by setting $w_i = \sum_{j=1}^d a^j_i B^j$ and $t=\sum_{j=1}^d l_iB^j$, where $B \geq mn$ is a power of two. This reduction clearly preserves whether the instance is a YES-instance as $B \geq mn$ prevents interaction between different segments of the bit-strings when integers are added.
\end{proof}

%% file: randomksum.tex
\section{Random \texorpdfstring{$k$}{K}-Sum}\label{sec:ksum}
We start by proving Theorem~\ref{thm:randksum} which we restate here for convenience and afterwards present the faster algorithm mentioned in Subsection~\ref{subsec:ourres} that assumes random read-only access to random bits.

\begin{reptheorem}{thm:randksum}
There is a randomized algorithm that given a constant $k \geq 2$, $k$ uniformly random vectors $w^1,\ldots,w^k \in_R [m]^n$, with $n \leq m \leq \poly(n)$ and $m$ being a multiple of $n$, and a target $t$ (that may be adversarially chosen depending on the lists), finds $s_i \in w^i$ satisfying $\sum_{i=1}^k s_i=t$ if they exist with constant probability using $\tilde{O}(n^{k-0.5})$ time and $O(\log n)$ space.
\end{reptheorem}

\begin{proof}

Let us first consider the case $k=2$. We will apply the List Disjointness algorithm from Section~\ref{sec:ld} with $s=1$ (so $K=\{k_1\}$), $x=w^1$, $y=w^2$ and $p=\Theta(n)$ (the latter is justified by an easy computation of $\mathbb{E}[p(w^1,w^2)]$)\footnote{Fix $v\in [m]$. Let $X_i=1$ if $w^1(i)=v$ and $0$ otherwise. Then $\E[|(w^1)^{-1}(v)|^2]=\sum_i \E[X_i]+2\sum_{i<j}\E[X_iX_j]=\frac{n}{m}+2\binom{n}{2}\frac{1}{m^2}$. Thus $\E[p(w^1,w^2)]=\sum_v\E[|(w^1)^{-1}(v)|^2]+\sum_v\E[|(w^2)^{-1}(v)|^2]=2n+4\binom{n}{2}\frac{1}{m}=\Theta(n)$.}. Recall that the List Disjointness algorithm assumes random read-only access to a random function $h:[m]\rightarrow [n]$, but here we do not use this assumption and show that we can leverage the input randomness to take 
\[
 h(v)=(v\bmod n)+1.
\]
Note that since $m$ is a multiple of $n$, we have that if $v \in_R [m]$, then both $h(v)$ and $h(t-v)$ are uniformly distributed over $[n]$.
To analyze this adjusted algorithm, we reuse most of the proof of Lemma~\ref{lem:LD-ST}. To facilitate this, we show that within one run of the loop of Line~\ref{lin:loop2} the following variant of Lemma~\ref{obseqd} holds, where $W_L(\{k_1\})$ is as defined in Section~\ref{sec:ld} (but with the new function $h$ as defined above).

\begin{lemma}\label{obseqd2}
Fix any value of integers $i^*,j^* \in [n]$ and $r$ such that $r$ is good. Let $x,y$ be random strings from $[m]^{n}$ such that $x(i^*)=y(j^*)$. Let $z' \in [m]^n$ be the vector with $z'(i)=i$ if $i \notin \{i^*,j^*\}$ and $z'(i)=n+1$ otherwise. Let $\beta$ be a random string which is generated by cutting off a uniform sample from the set $[n]^L$ at the first repetition with respect to $z'$.
Then for every string $\rho$ containing $i^*$ and $j^*$ it holds that $\Pr_{k_1,x,y}[W_L(\{k_1\})=\rho]\geq\Pr_{\beta}[\beta=\rho]$.
\end{lemma}
Intuitively, this lemma states that random instances in which $i^*$ and $j^*$ form a solution behave similarly as the algorithm from Lemma~\ref{lem:LD-ST} with the specific vector $z'$ in which $i^*,j^*$ are the only indices with common values.
\begin{proof}
	We may assume $\rho(i) \neq \rho(j)$ for every $i \neq j$ as otherwise $\Pr_{\beta}[\beta=\rho]=0$ since the only repetition with respect to $z'$ in $\rho$ must be formed by $i^*\neq j^*$. By the same argument we may also assume that the last entry of $\rho$ is $i^*$ or $j^*$. Denoting $L'$ for the length of $\rho$, note that $\Pr_{\beta}[\beta=\rho] \leq n^{-L'}$ as a necessary condition for $\beta=\rho$ is that the first $L'$ locations of the infinite random string used to construct $\beta$ match $\rho$. For notational convenience, let us denote $\alpha=W_L(\{k_1\})$. We have that
	\begin{equation}\label{eq:probchain}
	\Pr_{k_1,x,y}[\alpha=\rho] = \prod_{i=1}^{L'} \Pr_{k_1,x,y}[\alpha_i=\rho_i | (\alpha_1,\ldots,\alpha_{i-1})=(\rho_1,\ldots,\rho_{i-1})].
	\end{equation}
	We see that $\Pr_{k_1}[\alpha_1 = \rho_1]=1/n$ as $\alpha_1=k_1$ and $k_1$ is uniformly distributed over $[n]$. For $1<i\leq L'$ we know $\alpha_{i-1}=\rho_{i-1} \neq \rho_j =\alpha_j$ for every $j < i-1$ and $i^*$ and $j^*$ cannot both occur in $\alpha_1,\ldots,\alpha_{i-1}$ thus both $x_{\alpha_{i-1}}$ and $y_{\alpha_{i-1}}$ are independent of $\alpha_1,\ldots,\alpha_{i-1}$ and uniformly distributed over $[m]$ as $x,y$ are random strings with $x_{i^*}=y_{j*}$. Depending on $r$ we either have that $\alpha_i=h(x_{\alpha_{i-1}})$ or $\alpha_i=h(x_{\alpha_{i-1}})$, but in both cases
	\[
		\Pr_{k_1,x,y}[\alpha_i=\rho_i | (\alpha_1,\ldots,\alpha_{i-1})=(\rho_1,\ldots,\rho_{i-1})] = 1/n,
	\]
	since $m$ is a multiple of $n$. Thus~\eqref{eq:probchain} equals $n^{-L'}$ and the lemma follows.
\end{proof}

As a consequence of this lemma we have that
\begin{equation}\label{eq:avginst}
\begin{aligned}
	\Pr_{k_1,x,y}[W_L(\{k_1\}) \text{ contains } i^*,j^* | r \text{ is good} \wedge x(i^*)=y(j^*)] &\geq \Pr_{\beta}[\beta \text{ contains } i^*,j^*]\\
	&\geq \Omega((L/n)^2),
\end{aligned}
\end{equation}

where $\beta$ is distributed as in Lemma~\ref{obseqd2} and we use Lemma~\ref{lem:lwbndprob} with $s=1$, $z=z'$ and $p=n+1$ for the second inequality. Note the probability lower bounded in~\eqref{eq:avginst} is sufficient for our purposes, as we assume $i^*,j^*$ exist in the theorem statement and Observation~\ref{obs:good} still holds. Using the proof of Lemma~\ref{lem:LD-ST}, the adjusted algorithm thus solves 2-Sum in time $\tilde{O}(n\sqrt{p})=\tilde{O}(n^{1.5})$ and $O(\log n)$ space with constant probability.

For $k$-Sum with $k >2$, we repeat the above for every tuple from the Cartesian product of the last $k-2$ sets. This blows up running time by a factor of $n^{k-2}$ and space by an additive factor of $k\log n$. It is easy to see that this algorithm will still find a solution with constant probability if it exists since it only needs to do this for the correct guess of the integers from the last $k-2$ sets.
\end{proof}

The following result follows more directly from our techniques from Section~\ref{sec:ld}. 

\begin{theorem}
Let $k,s,m$ be integers such that $k$ is even, $k \geq 2$, $s\leq n^{k/2}$ and $n^{k/2} \leq m \leq \poly(n)$. There is a randomized algorithm that given $k$ uniformly random vectors $w^1,\ldots,w^{k} \in_R [m]^n$ and a target $t$ (that may be adversarially chosen depending on the lists) finds $s_i \in w^i$ satisfying $\sum_{i=1}^{k} s_i=t$ if it exists with constant probability using $\tilde{O}(n^{3k/4}/\sqrt{s})$ time and $O(s\log n)$ space. The algorithm assumes random read-only access to a random function $h:[km] \rightarrow [n^{k/2}]$.
\end{theorem}
\begin{proof}
	We may assume $t \leq km$ since otherwise the answer is trivially NO. Define two lists
	\[
		x = \left(\sum_{i=1}^{k/2} w^i_{s_i}\right)_{(s_1,\ldots,s_{k/2}) \in [n]^{k/2}} \qquad \qquad y= \left(t- \sum_{i=1}^{k/2} w^{k/2+i}_{s_i}\right)_{(s_{1},\ldots,s_{k/2}) \in [n]^{k/2}}.
	\] 
	Then for any value $v \in [m]$ and $v' \in \{t - m,\ldots,t\}$ we have 
	\[
		\E[|x^{-1}(v)|^2] \le n^{k/2}/m+n^{k}/m^2 \qquad \qquad \E[|y^{-1}(v')|^2] \le n^{k/2}/m+n^{k}/m^2.
	\]
	Thus we see that $\E[p(x,y)] \leq O(n^{k/2})$ so we may set $p=O(n^{k/2})$ and it will be a correct upper bound on the number of pseudo-collisions with constant probability by Markov's inequality. The algorithm then follows from Theorem~\ref{thm:randksum} as it runs in time $O(n^{3k/4}/\sqrt{s})$.
\end{proof}

%% file: conclusions.tex
\section{Further Research}\label{sec:conc}

Our work paves the way for several interesting future research directions that we now briefly outline.

\paragraph{The random read-only access to random bits assumption:}
An important question is whether the random access read-only randomness assumption is really required to improve over exhaustive search. Note that this assumption is weaker than the assumption that non-uniform exponentially strong pseudorandom generators exist (see e.g.~\cite{DBLP:conf/tcc/JainPT12}). Specifically, if such pseudorandom generators exist they could be used in the seminal construction by Goldreich et al. (see e.g.~\cite[Proposition 3]{DBLP:conf/tcc/JainPT12}) to build strong pseudorandom functions $h: [m] \rightarrow [n]$ required for the algorithm listed in Figure~\ref{alg:LD-ST}: when the algorithm fails to solve a particular instance of List Disjointness, this instance could be used as advice to build a distinguisher for the pseudorandom generator (this is similar to the proof of other derandomization results, see e.g.~\cite[Section 9.5.2 and 20.1]{DBLP:books/daglib/0023084}).
Interestingly, Impagliazzo et al.~\cite{DBLP:journals/joc/ImpagliazzoN96} show that if Subset Sum is sufficiently \emph{hard} in the average case setting, the Subset Sum function that computes the sum of a given subset is a good candidate for a pseudorandom function. 
Beame et al.~\cite{beame2013element} also raise the question whether the random access read-only randomness assumption is really required and in particular ask whether $\mathrm{polylog}(n)$-wise independent randomness might be useful to this end.

\paragraph{Solving instances of List Disjointness with many pseudo-collisions faster:}
A natural question is whether our dependence in the number of pseudo-collisions is really needed. Though this number comes up naturally in our approach it is somewhat counterintuitive that this number determines the complexity of an instance. It would be interesting to find lower bounds, even for restricted models of computation, showing that this dependence is needed, or show the contrary.

\paragraph{Other applications:}
Due to the basic nature of the List Disjointness problem, we expect there to be more applications of Theorem~\ref{thm:LD-ST}. Note that already the result of Beame et al.~\cite{beame2013element} implies space efficient algorithms for e.g. the Colinear problem (given $n$ points in the plane are $3$ of them on a line). In the area of exact algorithms for NP-complete problems, it is for example still open to solve MAX-2-SAT or MAX-CUT in $O^*((2-\varepsilon)^n)$ time and polynomial space (see~\cite[Section 9.2]{DBLP:series/txtcs/FominK10}) or the Traveling Salesman problem in time $O^*((4-\varepsilon)^n)$ time and polynomial space (see~\cite[Section 10.1]{DBLP:series/txtcs/FominK10}) for $\varepsilon >0$. Other applications might be to find space efficient algorithms to check whether two vertices in an $n$-vertex directed graph with maximum (out/in)-degree $d$ are of distance $k$ from each other: this can be done in time $O(d^{\lceil k /2 \rceil})$ time and space or in $O(d^k)$ time and $O(k \log n)$ space. Improving the latter significantly for $d=2$ would generalize Theorem~\ref{thm:sss}.
